\documentclass[11pt,a4]{llncs}
\usepackage[T1]{fontenc}

\usepackage[font=small]{caption}
\usepackage{latexsym}
\usepackage{graphicx}
\usepackage{tikz}
\usepackage{color}
\usepackage{url}
\usepackage{subfig}
\usepackage{amsmath,amssymb}
\usepackage{amsfonts}

\usepackage{mathtools}

\newcommand{\stirlingtwo}[2]{\genfrac{\lbrace}{\rbrace}{0pt}{}{#1}{#2}}

\begin{document}

\makeatletter
\renewcommand\paragraph{\@startsection{paragraph}{4}{\z@}%
                   {-8\p@ \@plus -2\p@ \@minus -2\p@}%
                   {-0.5em \@plus -0.22em \@minus -0.1em}%
                   {\normalfont\normalsize\itshape}}
\makeatother

\title{Counting small cuts in a graph} 
\author{
  Barbara Geissmann
  \and
  Rastislav \v{S}r\'amek}
\institute{Institute of Theoretical of Computer Science, ETH Zurich, Zurich, Switzerland\\
  e-mail:~geibarba@student.ethz.ch, rsramek@inf.ethz.ch}

%\begin{titlepage}
\maketitle
\begin{abstract}
We study the minimum cut problem in the presence of uncertainty and
show how to apply a novel robust optimization approach, which aims 
to exploit the similarity in subsequent graph measurements or similar 
graph instances, without posing any assumptions on the way they have 
been obtained. With experiments we show that the approach works well 
when compared to other approaches that are also oblivious towards the
relationship between the input datasets. 
\end{abstract}

%\end{titlepage}
%\setcounter{page}{1}
\section{Introduction}

Dealing with uncertainty is an ever more common problem. We are flooded with data 
recorded by virtually all modern devices from cars to cellular phones, data about 
various networks, observations of different phenomena. 
In order to be able to extract meaningful information 
from this data, we need to be able to remove or at least identify the noise that is inherently 
present, whether due to measurement errors or due to systematic influence
of unknown factors. 
In this paper we consider a novel method of robust optimization introduced by Buhmann et al.
\cite{buhmann2013}, and apply it to the problem of searching for the global minimum cut in a graph. 

%Most approaches of dealing with noisy inputs can be roughly classified into one 
%of two classes. One is to model the data generator and try to estimate it's 
%noise-less parameters, the other is to take the data as given and try to devise 
%a problem solution that is general enough to work for the noise-less instance 
%as well. In this paper we will consider one of the latter methods of robust optimization, 
%introduced by Buhmann et al. \cite{buhmann2013}, and apply it to the problem of searching 
%for the global minimum cut in a graph. 

Finding the global minimum cut in a graph is a well studied problem with applications 
ranging from information retrieval \cite{botafogo1993} to computer vision 
\cite{boykov2004}. The problem is to separate the set of graph vertices $V$ into two non-empty 
disjoint sets $X$ and $X\setminus V$, such that the sum of the weights of edges that have one end-point 
in $X$ and another in $X\setminus V$ is minimized. Since a cut is fully determined by the subset $X$, 
we will denote it only by $X$ with the possible caveat that $X$ and $V\setminus X$ denotes 
the same cut.
We are in particular interested in the minimum cut as a measure 
of network robustness \cite{ramanathan1987}: If the weight of an edge represents the effort needed to cut 
that particular edge, the minimum cut represents the least effort necessary to disconnect 
the graph.

Suppose that we are looking for a minimum cut in a graph, for instance one that represents connections 
between nodes in a sensor network. However, instead of the ``true'' graph we are only given 
two snapshots of it from two different points in time, with the same topology, but with different 
edge-weights. What should we do in order to identify a minimum cut in the ``true'' underlying graph, 
or a cut that will be minimum in a third, similar, snapshot? 

It is clear that without very precise understanding of the process by which we obtain the 
graph measurements, we are unable to answer this question with full confidence and thus any solution 
will be only an heuristic. Nevertheless, 
the setting is a realistic and very common one, and we should not give up. For instance, one intuitive course 
of action would be to average the weights provided by two instances edge by edge and compute a minimum cut 
on the resulting 
graph. In this paper we want to show that a different method, also oblivious to the properties of the 
data generator, might yield better results.

\subsection{Approximation set optimization}
We will introduce the aforementioned robust optimization method of Buhmann et al. \cite{buhmann2013} 
in greater detail. We will refer to it as approximation set optimization. Recall that the weight 
of a graph cut $X$ is the sum of the weights of the edges that have one endpoint in $X$ and another 
in $V\setminus X$, we will denote it by $w(X)$. 

\begin{definition}[$\rho$-Approximate Cut]\\
Let $\lambda(G)$ denote the weight of a global minimum cut in $G$. For a parameter $\rho \geq 1$, 
a $\rho$-approximate cut $X$ is a cut with weight at most $\rho\lambda(G)$, $w(X)\leq\rho\lambda(G)$.
\end{definition}

\begin{definition}[$\rho$-Approximation Set]\\
A $\rho$-approximation set of $G$, denoted by $A_\rho(G)$, is the set of all $\rho$-approximate cuts in $G$,
$A_\rho(G) = \{X\in V| w(X) \leq \rho\lambda(G)\}$. 
\end{definition}

Let $G_1$ and $G_2$ be two weighted graphs with the same topology but different edge-weights. The approximation 
set optimization method states that we should find a factor $\rho$, for which the intersection of 
the $\rho$-approximation sets $A_\rho(G_1)\cap A_\rho(G_2)$ is the largest, when compared to the expected 
size of this intersection if the instances were generated at random. We then pick a solution at random from 
the intersection of the resulting $\rho$-approximation sets. Formally, we look for $\rho^*$ such that 
\begin{equation}
\rho^*=\arg\max_\rho{|A_\rho(G_1)\cap A_\rho(G_2)| \over Es(|A_\rho(G_1)|,|A_\rho(G_2)|)},
\label{eq:similarity}
\end{equation}
where $Es(|A_\rho(G_1)|,|A_\rho(G_2)|)$ is the expected size of the intersection of the $\rho$-approximation sets 
of the given size. We call the value 
\begin{equation}|A_{\rho^*}(G_1)\cap A_{\rho^*}(G_2)|/Es(|A_{\rho^*}(G_1)|,|A_{\rho^*}(G_2)|)
\end{equation} \textit{unexpected similarity}. It
is a measure of similarity of $G_1$ and $G_2$, with respect to the optimization problem of looking for the 
minimum cut. 
%Finally, we and pick a random solution from the intersection of
%$A_{\rho^*}(G_1)$ and $A_{\rho^*}(G_2)$. This solution is expected to be minimal for any further graph $G_3$.

In order to successfully apply the method, we need to be able to solve five problems: 
Count the number of $\rho$-approximate cuts in a graph $G$, 
count the number of cuts in the intersection of the approximate sets of two graphs $G_1$ and $G_2$, 
compute the function for the expected intersection $Es$, find the optimal factor $\rho^*$, and 
choose a cut at random from the set of all cuts that are $\rho$-approximate for the graphs $G_1$
and $G_2$ at the same time.

\subsection{Related work}

Robust optimization is a widely studied subject. However, in order to be able to derive provably optimal methods, 
one needs to restrict the scope of inputs, or make other strong assumptions about them. For instance 
\emph{Stochastic optimization} \cite{schneiderstochastic,Kall:Mayer} and \emph{Robust optimization} 
\cite{BEN:09} expect that we know respectively the complete distribution of an instance and the 
complete set of instances. Various methods of optimization for stable inputs on the other hand suppose 
that the input cannot change too much \cite{linial2010,bilo08,mihalak2011}. In our case, by not assuming 
anything about the input we lose the ability to apply any of these but
greatly increase the scope of problems for which we can hope to achieve good solutions.

The remainder of the paper is structured as follows. In Section \ref{sec:algorithm} we show how to count the sizes of $\rho$-approximation sets of cuts and their intersections, in Section \ref{sec:expectation} we will derive a approximate formula for 
the expected size of the approximation set intersection on random instances, followed by experimental evaluation of the method 
in Section \ref{sec:experiments} and concluding remarks in Section \ref{sec:conclusion}.

\section{Algorithms for counting small cuts \label{sec:algorithm}}

For many combinatorial optimization problems, the problem of counting approximate solutions is \#P-complete, even if the 
optimization problem itself is efficiently solvable. The reason for this lies in the possibly exponential number of 
solutions. For instance, there can be $n^{n-2}$ short spanning trees in a graph with $n$ vertices or $2^{n-2}$ 
short $s$-$t$ paths in a directed acyclic graph \cite{mihalak2013}. For minimum cuts, 
however, the possible number of near-optimal cuts is small. Dinits et al. \cite{dinits} showed that there can be at 
most ${n \choose 2}=O(n^2)$ minimum cuts in a graph and Karger \cite{kargerbounds} showed that the number of 
$\rho$-approximate cuts is at most $O(n^{2\rho})$. This makes our life significantly easier, since we can afford to 
enumerate, not only count the cuts in the approximation sets. Note that calculating the number of cuts shorter than 
an arbitrary threshold is still \#P-complete \cite{provan1983}. This is not surprising, since with rising threshold 
the problem must turn from easy to difficult, as calculating the maximum graph cut is a NP-complete problem. 

There are at least two different algorithms that can compute the $\rho$-approxi\-mation sets of a graph. 
One is by Nagamochi, Nishimura and Ibaraki \cite{nagamochi} and it solves the task deterministically in 
$O(mn^{2\rho})$ time if $m$ is the number of edges of the graph with $n$ vertices. The other is an adaption 
of the \textit{recursive contraction algorithm} 
by Karger and Stein \cite{kargerstein}, and it finds all $\rho$-approximate cuts in $O(n^{2\rho}\log^3n)$ 
time with high probability.

We will use the approach of Karger and Stein because it is the fastest currently known algorithm. Apart from 
that it allows us to make an adaptation with which we can directly compute the approximation sets.

\subsection{Karger and Stein's Algorithm}
The recursive contraction algorithm, described in Algorithm 1, finds a minimum cut in a graph as follows. 
The main idea is to repeatedly choose an edge at random and contract it, which means that the two end 
vertices of this edge are merged into a single vertex. The algorithm starts with two copies of the graph. 
On each of them it performs random edge contractions until the graph has shrunk down to a certain size.
Then the graph is copied again and the algorithm continues, again on both graphs. When only two vertices 
remain, the edges contracted into each of the two vertices correspond to one set of vertices in a cut and 
the weight of the remaining edge corresponds to the cost of this cut. The algorithm keeps track of the 
found cuts and the best cut is returned. The intuition behind the algorithm is that in the beginning, the probability 
of contracting an edge from a minimum cut, and thus excluding this cut from the set of possible results, is low. 
As the algorithm progresses, this chance increases, but this is combated by the increased number of concurrent 
evaluations. 

\begin{algorithm}[Recursive Contraction Algorithm]
\begin{tabbing}
\hspace{0cm} \= \hspace{.6cm} \= \hspace{.6cm} \= \kill
\>\textsc{RecursiveContract}($G$)\\
\>\textbf{if }$\vert V \vert \leq 6$ \textbf{then}\\
\>\>$G \leftarrow$ \textsc{Contract}($G$, $2$)\\
\>\>\textbf{return} the cut\\
\>\textbf{else}\\
\>\>\textbf{repeat twice}\\
\>\>\>$G' \leftarrow$ \textsc{Contract}($G$, $\lceil n/\sqrt{2}+1\rceil$)\\
\>\>\>\textsc{RecursiveContract}($G'$)\\
\>\>\textbf{return} the smaller cut\\
\>\textbf{end}
\end{tabbing}
\end{algorithm}

The routine \textsc{Contract}($G$, $x$) does repeated edge contraction in $G$ until only $x$ vertices remain. 
The whole algorithm runs in $O(n^2 \log n)$ time. The probability that it finds a particular minimum cut is at 
least $\Omega(1/\log n)$. If we repeat the algorithm $O(\log^2 n)$ times, we will find any particular minimum cut
with high probability. The algorithm can be adjusted so that it returns all minimum cuts that it finds instead of 
only one. Since the total number of unique minimum cuts in a graph is bounded from above by ${n \choose 2}$, 
we can find every minimum cut with high probability, within the total time complexity of $O(n^2 \log^3 n)$.

Karger and Stein's algorithm can be modified to find all $\rho$-approximate cuts \cite{kargerstein}, 
by changing the reduction factor from $\lceil n/\sqrt{2}+1\rceil$ to
$\lceil n/\sqrt[2\rho]{2}+1\rceil$ and stopping the contraction when $2\rho$ 
vertices remain. In this case, all remaining possible cuts are evaluated. 
The running time increases to $O(n^{2\rho} \log n)$, whereas the success 
probability remains the same. Since the number of $\rho$-approximate cuts is 
bounded by $\Theta(n^{2\rho})$, we can find
all $\rho$-approximate cuts with high probability by repeating the algorithm 
$O(\log^2 n)$ times. This gives the overall time of $O(n^{2\rho}\log^3n)$.

\subsection{Approximation Set Optimization Algorithm}
Recall that we want to determine $\rho^∗$ that maximizes the unexpected similarity of the two graphs $G_1$ and $G_2$ with respect to the minimum cut 
problem. To this end we first need to compute the $\rho$-approximation sets of $G_1$ and $G_2$ and their intersection. The former is done by \textsc{ApproximationSet}($G$), which is an adapted version of Karger and Stein's recursive contraction algorithm. Afterwards, we are ready to compute the expected  and unexpected similarity, $Es$ and u\_sim. By sampling for the best $\rho$ we
find the intersection of $A_{\rho^*}(G_1)$ and $A_{\rho^*}(G_2)$ from which we can pick a cut at 
random, as a solution that generalizes for both instances. The whole process is described in Algorithm \ref{algo:approx}.

\begin{algorithm}[Approximation Set Optimization Algorithm]\label{algo:approx}
\begin{tabbing}
\hspace{.6cm} \= \hspace{.6cm} \= \kill
\textbf{for } $\rho = 1$ \textbf{ until } $\rho = \text{MAX}$ \textbf{ do}\\
\>$A_\rho(G_1) \leftarrow \textsc{ApproximationSet}(G_1, \rho)$\\
\>$A_\rho(G_2) \leftarrow \textsc{ApproximationSet}(G_2, \rho)$\\
\>intersection $\leftarrow$ $intersect(A_\rho(G_1),A_\rho(G_2))$\\
\>u\_sim $\leftarrow$ $|$intersection$|$ $/$ $Es(|A_\rho(G_1)|,|A_\rho(G_2)|)$\\
\>\textbf{if } u\_sim $>$ max\_sim \textbf{then}\\
\>\>max\_sim $\leftarrow$ u\_sim\\
\>\>$\rho^*$\_intersection $\leftarrow$ intersection\\
\>\textbf{end}\\
\textbf{end}
\end{tabbing}
\end{algorithm}

Now let us discuss some issues of the algorithm in more detail and derive its time complexity.
Karger and Stein's version of the algorithm returns the cuts in an implicit way. Since we want to be able to compute the intersection of the approximation sets of two different graphs as well as to choose a cut from the intersection and apply it to a third graph, we need them explicitly. One simple possibility to meet this requirement is to store for each vertex whether it is in the cut or not. The entire cut can then be represented as a bit string of length $n$. Notice, that this notation is ambiguous, since the inverse of a bit string describes the same cut. We can fix this by allowing only cuts that have the first bit set to zero.

By treating the bit strings as numbers, we can sort the cuts in the approximation sets and then build the intersection in a merging fashion in $O(n^{2ρ} \log n)$ time, since the number of cuts in each approximation set is bounded from above by $\Theta(n^{2\rho})$.

Returning the cuts in an explicit manner also implies extra work during the computation of the approximation sets. After every recursion phase the union of the two found approximation sets is returned. To overcome difficulties like different smallest cut weights and duplicates, one possibility is to again sort the cuts. This extra work requires a factor of $O(\log n)$ additional time for the entire approximation set algorithm. So we end up with an approximation set algorithm that takes $O(n^{2\rho}\log^4n)$ time.

The time to compute u\_sim, the unexpected similarity, depends on the complexity of the function $Es$. We postpone this to Section 3.

The last thing we have to look at is the range and step size of the values for $\rho$ in the for loop. To choose a good bound for the
largest $\rho$ we want to test is not easy. It depends a lot on the structure of the graphs and the range of their weights. Therefore, we may want to start with rather big steps and refine them as we go.

\section{Expected intersection size \label{sec:expectation}}

Having described an algorithm that counts the size of individual approximation sets and 
their intersection, we turn to the question of deriving a formula for the expected size of the
intersection of the approximation sets. For a more detailed exposure we refer the reader to 
Chapter 3.2 of the bachelor thesis of Barbara Geissmann \cite{thesis}. 

For the expected similarity, we will only consider cuts on complete graphs. Otherwise we would need 
to track whether each cut $X$ cuts the graph into only two parts, since the Karger-Stein algorithm 
and its modifications return only such cuts. 

We first show that an arbitrary subset of cuts does not necessarily have to form a valid approximation set.
\begin{definition}[Crossing Cuts]
Two cuts $X$ and $Y$ cross each other if $X\cap Y \neq \emptyset$,
$X-Y\neq\emptyset$, $Y-X\neq\emptyset$, and $V-X-Y\neq\emptyset$. 
\end{definition}

\begin{definition}[Composed Cuts]\label{def:compcuts}
Let $X$ and $Y$ be two cuts that cross each other. Then they must define four further
cuts: 

\begin{equation}
\begin{array}{lll}
Z_1=X\cap Y & \qquad\qquad & Z_2=X-Y\\
Z_3=Y-X & & Z_4=V-X-Y
\end{array}
\end{equation}
We call $Z_1$, $Z_2$, $Z_3$, and $Z_4$ the \emph{composed cuts} of $X$ and $Y$. 
\end{definition}

\begin{theorem}
If two cuts $X$ and $Y$ in the approximation set $A_\rho(G)$ cross each other, 
then at least two of the four composed cuts of $X$ and $Y$ have to be in $A_\rho(G)$ as well.
\label{thm:cuts}
\end{theorem}

\begin{proof}
According to the Figure \ref{fig:cuts} we denote by $a$, $b$, $c$, and $d$ the sums of the 
weights of the cut edges between the sets $Z_1$, $Z_2$, $Z_3$, and $Z_4$, as in Definition \ref{def:compcuts}.
%$X-Y$ and $V-Y$, $Y-X$ and $X$, $X-Y$ and $Y$, and $Y-X$ and $V-X$ respectively. 
Without loss of generality, let us suppose that $a\leq b$, 
$c\leq d$, $b\leq d$. Then, for cuts $X$ and $Y$ to be in the approximation 
set, there must be a threshold $t:=\rho\lambda(G)$ such that $a+b\leq t$ and $c+d\leq t$. The $4$ 
composed cuts will have weights $a+c$, $a+d$, $b+c$, and $b+d$. However, it must hold that 
$a+c\leq t$ because both $a$ and $c$ are at most $t/2$, and also $b+c\leq t$  because we can 
replace $d$ in $c+d\leq t$ with $b$ which is at most as large. 
\qed
\end{proof}

\begin{figure}[tb]
\begin{center}
   \includegraphics[width=.48\textwidth]{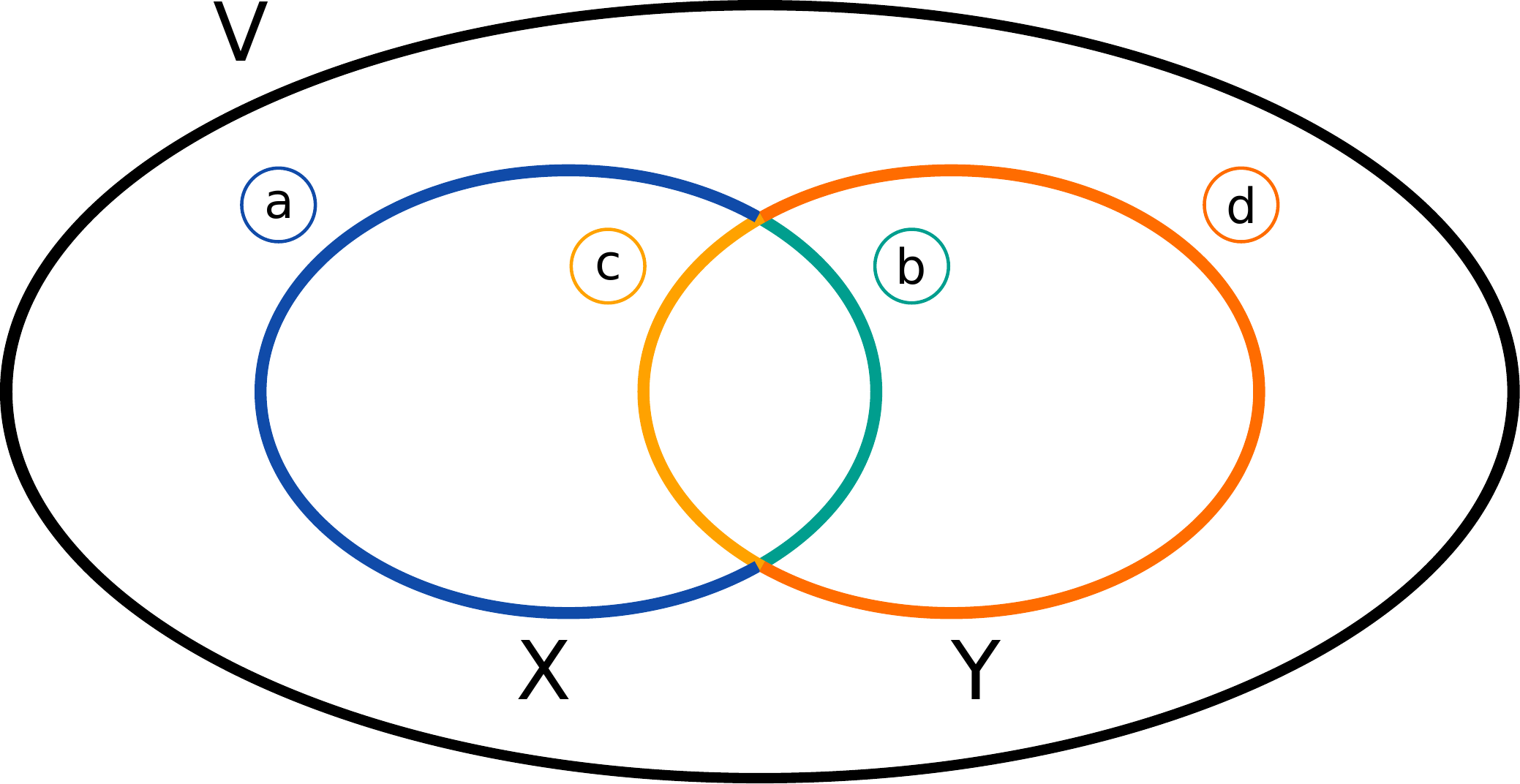}
  \caption{Crossing cuts\label{fig:cuts}}
\end{center}
\end{figure}

Theorem \ref{thm:cuts} shows that not every subset of cuts forms a feasible approximation set. 
Using Theorems 1 and 2 from \cite{buhmann2013} we can conclude that the expression 
$|A_\rho(G_1)||A_\rho(G_2)|/|S|$, where $S$ denotes the set of all cuts, is a lower bound
on the expected size of the intersection, but not its true value. 

We see that as soon as we have crossing cuts, we lose freedom in the number of cuts which we can freely choose.
We will show that this loss of freedom is substantial enough that by restricting ourselves to approximation sets without 
crossing cuts we get a approximation of the true expected value.

%In order to get a tighter lower bound, 
%we restrict ourselves to approximation sets without crossing cuts. Our intuition in this is based 
%on the Theorem \ref{thm:cuts}. In other words, if the number of vertices $n$ is large compared to the number of cuts in 
%the approximation set, every two cuts that cross reduce the number of cuts which we can choose freely by at least two. 
%Observe that in a complete graph, a non-empty cut on $n$ vertices can be chosen in $2^{n-1}-1$ ways. This implies that the loss of freedom to choose $2$ additional cuts significantly outweighs the additional flexibility we gained by choosing the second cut in $2^{n-1}-2$ ways. Furthermore we notice that with increasing number of crossing cuts, the number of composed cuts grows even further.

The number of ways in which we can cut a graph of $n$ vertices $m$ times so that the cuts do not cross 
is equal to the number of ways we can partition a set of $n$ integers into $m+1$ non-empty subsets. 
The latter describes the well known Stirling number of the second kind, denoted by $\stirlingtwo{n}{m+1}$, 
and defined by the explicit formula
\begin{equation}
\stirlingtwo{n}{k}={1\over k!}\sum_{j=0}^k(-1)^{k-j}{k \choose j}j^n.
\end{equation}
Observe that in a complete graph, a non-empty cut on $n$ vertices can be chosen in $2^{n-1}-1$ ways. Furthermore, the smallest approximation set that can contain crossing cuts is of size $4$. Such an approximation set would be approached by $\stirlingtwo{n}{5}$.
 We conclude that at least if the number of vertices $n$ is large compared to the number of cuts in 
the approximation set, the loss of freedom to choose $2$ additional cuts significantly outweighs the additional flexibility we gained by choosing the second cut in $2^{n-1}-2$ ways. Note that with increasing number of crossing cuts, the number of composed cuts grows even further.

We now calculate the expected similarity for non-crossing cuts and use it as an approximation for the unexpected 
similarity when cuts cross. 

Let $k:=|A_\rho(G_1)|$, $l:= |A_\rho(G_2)|$, and $\mathcal{F}_x$ denote all approximation sets that contain 
$x$ cuts. Then by Lemma 2 of \cite{buhmann2013} we have 
\begin{eqnarray*}
Es(k,l) &=& \frac{1}{|\mathcal{F}_k||\mathcal{F}_l|} \sum_{\substack{F_1\in \mathcal{F}_k\\F_2\in
    \mathcal{F}_l}} |F_1\cap F_2| = 
\frac{1}{|\mathcal{F}_k||\mathcal{F}_l|} \sum_{s\in S}
    |\{F\in \mathcal{F}_k | s\in F\}|\cdot|\{F\in \mathcal{F}_l | s\in F\}| \\
&=& \frac{1}{\stirlingtwo{n}{k+1}\stirlingtwo{n}{l+1}}\cdot \frac{\sum_{i=1}^{n-1}\left({n \choose i} \sum_{j=0}^{k-1}\left( \stirlingtwo{i}{j+1}\stirlingtwo{n-i}{k-j} \right) \cdot \sum_{j=0}^{l-1} \left( \stirlingtwo{i}{l+1}\stirlingtwo{n-i}{l-j} \right) \right)}{(2^{k+1}-2)(2^{l+1}-2)}.
\end{eqnarray*}

The factors $2^{k+1}-2$ and $2^{l+1}-2$ prevent from double-counting by choosing the same cuts in a different order.

Deriving a closed formula for $Es$ seems difficult, due to the Stirling numbers of the second kind. We can, however, 
evaluate the expression algorithmically for every necessary $k$ and $l$. In order to avoid straight-forward $O(n^7)$ computation, 
we pre-compute all binomial coefficients from ${n \choose 1}$ to ${n \choose n-1}$ in linear time using the identity
${n \choose i+1} = {n \choose i}\cdot \frac{(n-i)}{(i+1)}$. 
%The same is true for a pre-computation of all powers of two. - We can do this in constant time by bit-shift!
Similarly, using the combinatorial identity $\stirlingtwo{n}{k} = k\stirlingtwo{n-1}{k} + \stirlingtwo{n-1}{k-1}$,
we can pre-compute all Stirling numbers from $\stirlingtwo{0}{0}$ to $\stirlingtwo{n}{n}$ in time $O(n^2)$. 
Using the previously pre-computed values, we can compute all inner summands for different values of $l$ and $k$ 
in $O(n^3)$ time and space. An evaluation of the formula for two particular values of $k$ and $l$ thus needs only
$O(n)$ time and the evaluation for all possible pairs of $k$ and $l$ thus fits into the $O(n^3)$ time necessary for
the pre-computations.

\section{Experimental results \label{sec:experiments}}
In order to evaluate the performance of the approximation set optimization for this problem, we tested it on various sets of input 
instances and compared the performance to two other algorithms. The first being an algorithm where we average edge weights edge 
by edge and compute minimum cut on the resulting graph and the second being an algorithm where we increase $\rho$ until the intersection 
of $\rho$-approximation sets is non-empty for the first time and we choose the cut from this intersection. We look at this second 
algorithm because it intuitively seems to be a very good approach. 
\subsection{Tests}
Every test is as follows. Three complete, undirected, weighted graph instances are taken as input, where the first two are used to predict a good solution for a future one. Then, this solution is tested against the third instance. Figure \ref{fig:test} illustrates all tests done.
\begin{figure}[htbp]

\textbf{Input:} Three complete, undirected, weighted graphs, $G_1$, $G_2$, and $G_3$.\\
\textbf{Output:} Four different results:
\begin{itemize}
\item \textsc{Average:} Add the edge weights of $G_1$ and $G_2$ pairwise. Compute a minimum cut for the new formed graph. Apply the solution on $G_3$.

\item \textsc{FirstIntersection:} Find the smallest $\rho$ that results in a non-empty intersection of the $\rho$-approximation sets for $G_1$ and $G_2$. Pick a random cut from the intersection and apply it on $G_3$.

\item \textsc{BestSimilarity:} Find $\rho^*$ which maximizes the unexpected similarity of $G_1$ and $G_2$. Pick a random cut from the intersection and apply it on $G_3$. 
%Additionally, return the average, the best and the worst cost of all cuts in the intersection applied on $G_3$.

\item \textsc{Optimum:} Compute a minimum cut of $G_3$.
\end{itemize}
\caption{Specification of the Experiment}
\label{fig:test}
\end{figure}

\subsection{Data}
We run experiments on three different kind of graphs to evaluate the quality of the found solution: On graphs constructed with real world data
which we expect to be similar, on totally random graphs which we do not expect to be similar at all, and on artificially generated similar 
random graphs, which all have some small cuts in common.
The tests on real world data are based on the historical daily prices between 1999 and 2010 of thirteen different stock 
indices\footnote{BEL-20, Dow Jones, Hang Seng, Nikkei, AEX, CAC-40, Dax, Eurotop100, FTSE100, JSX, Nasdaq, AS30, RTSIndex, SMI} 
\cite{stock}. The vertices of our graph correspond to individual stock indices and the edges between them correspond to their 
similarity with respect to the problem of finding a contiguous sub-array of maximum sum\footnote{In other words, finding out when 
to buy and when to sell in order to maximize profit, if we are only allowed to do each operation once.}, as calculated by the 
approximation set optimization method
 \cite{buhmann2013}. Every graph corresponds to one year.
For the random graphs, we assign a random weight to every edge.
For the artificially made similar graphs we randomly define some cuts to be small and allocate small weights to their edges. To all the other 
edges we randomly assigned a weight from a larger range.

\subsection{Results}
\paragraph{Real World Data.}
To overcome the need of sampling for very high values of $\rho$, we took logarithms of the edge weights. In most of the 
tests $\rho^*$ was thus smaller than $3.0$. The results are listed in Figure \ref{fig:real_world_data_log}. 
In addition to results on all tests, we extracted pairs of instances with higher than median unexpected similarity 
and tried to use only those to predict results. As perhaps the only unexpected result, this did not seem to improve the 
specificity. It seems that the differences between various years vary too much (which corresponds to our ability to predict 
market behavior, which is, in general, poor). 

\begin{figure}[htbp]
	\centering
	\begin{tabular}{l|r|r|r|r}
		&sum of&\% of opt&sum of all&\% of opt\\
		&all tests&&tests with&\\
		&(858)&&$U \geq \tilde{U}$ (462)&\\
		\hline
		Average & 65062.20 & 188.70\% & 34826.08 & 187.28\% \\
		First Intersection & 63682.42 & 184.70\% & 34702.42 & 186.62\% \\
		Best Rho & 63116.60 & 183.06\% & 34702.42 & 186.62\% \\
		Optimum & 34478.63 & 100.00\% & 18595.24 & 100.00\%
	\end{tabular}
	\caption{Stock Market Data (Logarithmised)}
	\label{fig:real_world_data_log}
\end{figure}

\paragraph{Random Graphs.}
These experiments are mainly done for control purposes. If the data is truly random, 
we do not expect any algorithm to hold a significant edge, and indeed, the results
reflect this. Note that while no algorithm works well here, we are able to realize
that this will be so due to low unexpected similarity between instances. 
Figure \ref{fig:random_result_15} and Figure \ref{fig:random_result_50} list the results. 

\begin{figure}[htbp]
	\centering
	\begin{tabular}{l|r|r|r|r}
		&sum of&\% of opt&sum of all&\% of opt\\
		&all tests&&tests with&\\
		&&&$U \geq \tilde{U}$ (260)&\\
		\hline
		Average & 460763 & 132.74\% & 232340 & 132.91\%\\
		First Intersection & 459802 & 132.47\% & 232688 & 133.11\% \\
		Best Rho & 458033 & 131.96\% & 232546 & 133.03\% \\
		Optimum & 347112 & 100.00\% & 174810 & 100.00\%
	\end{tabular}
	\caption{Random Graphs of 15 Vertices, Edge Weight Range [0-255]}
	\label{fig:random_result_15}
\end{figure}

\begin{figure}[htbp]
	\centering
	\begin{tabular}{l|r|r|r|r}
		&sum of&\% of opt&sum of all&\% of opt\\
		&all tests&&tests with&\\
		&(512)&&$U \geq \tilde{U}$ (261)&\\
		\hline
		Average & 1616070 & 122.16\% & 820594 & 121.99\% \\
		First Intersection & 1612010 & 121.86\% & 821715 & 122.15\% \\
		Best Rho & 1602646 & 121.15\% & 820574 & 121.99\%\\
		Optimum & 1322892 & 100.00\% & 672683 & 100.00\%
	\end{tabular}
	\caption{Random Graphs of 50 Vertices, Edge Weight Range [0-255]}
	\label{fig:random_result_50}
\end{figure}

\paragraph{Similar Random Graphs.}
With these experiments we wanted to verify our expectation that results improve with increasing similarity of graphs, e.g. 
the larger the expected value of a random cut gets compared to the expected value of a small cut, the better are our results,
see Figures \ref{fig:similar_result1}, \ref{fig:similar_result2}, and \ref{fig:similar_result3}. For fixed small cut cost we 
get even better results, see Figures \ref{fig:similar_result4}, \ref{fig:similar_result5}, and \ref{fig:similar_result6}.

\begin{figure}[htbp]
	\centering
	\begin{tabular}{l|r|r|r|r}
		&sum of&\% of opt&sum of all&\% of opt\\
		&all tests&&tests with&\\
		&(512)&&$U \geq \tilde{U}$ (259)&\\
		\hline
		Average & 142105 & 114.43\% & 65818 & 108.95\%\\
		First Intersection & 139536 & 112.36\% & 64596 & 106.93\%\\
		Best Rho & 139331 & 112.20\% & 64596 & 106.93\%\\
		Optimum & 124182 & 100.00\% & 60410 & 100.00\%
	\end{tabular}
	\caption{Similar Graphs with Small Range [0,31] and Big Range [0,255]}
	\label{fig:similar_result1}
\end{figure}

\begin{figure}[htbp]
	\centering
	\begin{tabular}{l|r|r|r|r}
		&sum of&\% of opt&sum of all&\% of opt\\
		&all tests&&tests with&\\
		&(512)&&$U \geq \tilde{U}$ (286)&\\
		\hline
		Average & 132521 & 116.03\% & 72361 & 113.60\% \\
		First Intersection & 131285 & 114.95\% & 70983 & 111.44\% \\
		Best Rho & 128573 & 112.57\% & 70983 & 111.44\% \\
		Optimum & 114213 & 100.00\% & 63697 &  100.00\%
	\end{tabular}
	\caption{Similar Graphs with Small Range [0,31] and Big Range [0,127]}
	\label{fig:similar_result2}
\end{figure}

\begin{figure}[htbp]
	\centering
	\begin{tabular}{l|r|r|r|r}
		&sum of&\% of opt&sum of all&\% of opt\\
		&all tests&&tests with&\\
		&(512)&&$U \geq \tilde{U}$ (258)&\\
		\hline
		Average & 126123 & 119.87\% & 61841 & 116.91\% \\
		First Intersection & 126831 & 120.54\% & 61923 & 117.06\% \\
		Best Rho & 123126 & 117.02\% & 61581 & 116.42\% \\
		Optimum & 105220 & 100.00\% & 52897 & 100.00\%
	\end{tabular}
	\caption{Similar Graphs with Small Range [0,31] and Big Range [0,63]}
	\label{fig:similar_result3}
\end{figure}

\begin{figure}[htbp]
	\centering
	\begin{tabular}{l|r|r|r|r}
		&sum of&\% of opt&sum of all&\% of opt\\
		&all tests&&tests with&\\
		&(512)&&$U \geq \tilde{U}$ (401)&\\
		\hline
		Average & 79607 & 110.85\% & 57501 & 104.94\% \\
		First Intersection & 78311 & 109.04\% & 57603 & 105.13\% \\
		Best Rho &  77953 & 108.54\% & 57573 & 105.08\%\\
		Optimum &  71818 & 100.00\%	 & 54792 & 	100.00\%
	\end{tabular}
	\caption{Similar Graphs with Small Cut Weight 240 and Random Weight Range [0,255]}
	\label{fig:similar_result4}
\end{figure}

\begin{figure}[htbp]
	\centering
	\begin{tabular}{l|r|r|r|r}
		&sum of&\% of opt&sum of all&\% of opt\\
		&all tests&&tests with&\\
		&(512)&&$U \geq \tilde{U}$ (402)&\\
		\hline
		Average & 155440 & 111.97\% & 117571 & 106.03\%\\
		First Intersection & 155676 & 112.14\% & 117530 & 106.00\%\\
		Best Rho & 153880 & 110.84\% & 117530 & 106.00\% \\
		Optimum & 138828 & 100.00\% & 110881 & 100.00\%
	\end{tabular}
	\caption{Similar Graphs with Small Cut Weight 500 and Random Weight Range [0,255]}
	\label{fig:similar_result5}
\end{figure}

\begin{figure}[htbp]
	\centering
	\begin{tabular}{l|r|r|r|r}
		&sum of&\% of opt&sum of all&\% of opt\\
		&all tests&&tests with&\\
		&(512)&&$U \geq \tilde{U}$ (374)&\\
		\hline
		Average & 292254 & 107.90\% & 217775 &  104.36\%\\
		First Intersection & 291094 & 107.47\% & 217787 & 104.36\%\\
		Best Rho & 288942 & 106.68\% & 217685 & 104.32\%\\
		Optimum & 270853 & 100.00\% & 208679 & 100.00\%
	\end{tabular}
	\caption{Similar Graphs with Small Cut Weight 1000 and Random Weight Range [0,255]}
	\label{fig:similar_result6}
\end{figure}

%In general, the experimental results reaffirm our expectation that the algorithm \textsc{BestRho}
%is best at generalizing among the three. In addition to this, the unexpected similarity 
%gives us additional information about the usefulness of our result. In some applications this 
%can be a signific

%\subsection{Discussion}
%In a nutshell, we can conclude that the approximation set algorithm for minimum cut works well on similar instances. 
%If the typical instances of a problem do not resemble at all, that means, if there is for instance too much noise, 
%then a high unexpected similarity of two graphs is no help to find a near-optimum solution for a general instance. 
%If, however, the typical instances are very similar then our algorithm can make use of this and by considering two 
%sample graphs find a solution which is close to optimum for every graph. Furthermore, we can state that for graphs 
%with high similarity, the differences between the solutions of \textit{First Intersection} and \textit{Best Rho} are 
%pretty small. Furthermore, if the edge weights differ too much from each other we can work with the logarithmic values 
%of the edge weights. This helps to reduce the minimum value of $\rho$ for which the intersection of the two approximation 
%sets is not empty. It also reduces the maximum value of $\rho$ we have to sample for which is crucial to keep the time 
%complexity of the algorithm small.

\section{Conclusion \label{sec:conclusion}}
%\subsection{Summary}

We showed how to apply approximation set optimization to the problem of looking for 
a minimum cut in a graph by adapting a known minimum cut algorithm and estimating 
the expected intersection of two sets of small cuts.

In general, the experimental results reaffirm our expectation that the algorithm
is better at generalizing than other simple heuristic algorithms. In addition to 
this, the unexpected similarity gives us additional information about the usefulness 
of our result. In some applications this can be a significant benefit. Having information 
about the quality of the calculated solution may be very important, in particular when 
the calculated solution is far from optimal.

By the choice of the optimal parameter $\rho$, our approach selects a set of minimum cuts
which are expected to have low weight in the following graph instances. This can be 
of significant help as it divides the solution space into sets of relevant and irrelevant
cuts, for instance, in a network robustness scenario, it separates the cuts that are likely 
to be critical from those that are not.

\bibliographystyle{unsrt}  
\bibliography{references}

\vfill\eject

\end{document}